\documentclass[a4paper,10pt]{article}
\usepackage{amsmath,amsfonts,amssymb,amsthm}
\usepackage{mathtools}
\usepackage{graphicx}
\usepackage{color}
\usepackage{amsmath}
\usepackage{comment}
\usepackage{amssymb}
\usepackage{eclbkbox}
\usepackage{url}
\usepackage{authblk}

\usepackage{a4wide}
\newtheorem{lemma}{Lemma}
\newtheorem{theorem}{Theorem}
\newtheorem{proposition}{Proposition}
\newtheorem{corollary}{Corollary}
\newtheorem{definition}{Definition}

\newcommand{\rw}{{\mathtt{rw}}}
\newcommand{\cw}{{\mathtt{cw}}}
\newcommand{\tw}{{\mathtt{tw}}}
\newcommand{\pw}{{\mathtt{pw}}}
\newcommand{\nd}{{\mathtt{nd}}}
\newcommand{\tc}{{\mathtt{tc}}}
\newcommand{\dc}{{\mathtt{cd}}}
\newcommand{\db}{{\mathtt{bd}}}
\newcommand{\mw}{{\mathtt{mw}}}
\newcommand{\vc}{{\mathtt{vc}}}
\newcommand{\dhd}{{\mathtt{dhd}}}
\newcommand{\cod}{{\mathtt{cod}}}
\newcommand{\rankw}{{\mathtt{rw}}}
\newcommand{\dd}{{\mathtt{d}}}

\newcommand{\TT}{{\mathcal T}}
\newcommand{\XX}{{\mathcal X}}
\newcommand{\EE}{{\mathcal E}}

\newcommand{\revised}[1]{\textcolor{black}{{#1}}}

\newcounter{mpproblem}[section]

\makeatletter
\newenvironment{mpproblem}[1]%
{%
    \protected@edef\@currentlabelname{#1}%
    \par\vspace{\baselineskip}\noindent%
    \ifx#1\empty %
    \else \refstepcounter{mpproblem}$($#1$)$ %
    \fi%
    \hfill%
    $\left|%
    \hfill%
    \hspace{0.00\textwidth}%
    \@fleqntrue\@mathmargin\parindent%
    \begin{minipage}{0.86\textwidth}%
    \vspace{-1.0\baselineskip}%
}%
{%
    \end{minipage}%
    \@fleqnfalse%
    \right.$%
    \par\vspace{\baselineskip}\noindent%
    \ignorespacesafterend%
}%
\makeatother

\newenvironment{mpproblem*}%
{%
    \begin{mpproblem}{}%
}%
{%
    \end{mpproblem}%
    \ignorespacesafterend%
}

\title{Computing \textsc{Densest $k$-Subgraph} with Structural Parameters\thanks{This work is partially supported by JSPS KAKENHI Grant Number JP19K21537, JP21H05852, JP21K17707, JP22H00513.}}
\author[1]{Tesshu Hanaka\thanks{\texttt{hanaka@inf.kyushu-u.ac.jp}}}

\affil[1]{Kyushu University, Fukuoka, Japan}

\date{}

\begin{document}

\maketitle

\begin{abstract}
\textsc{Densest $k$-Subgraph} is the problem to find a vertex subset $S$ of size $k$ such that the number of edges in the subgraph induced by $S$ is maximized.
In this paper, we show that \textsc{Densest $k$-Subgraph} is fixed parameter tractable when parameterized by neighborhood diversity, block deletion number, distance-hereditary deletion number, and cograph deletion number, respectively.
Furthermore, we give a $2$-approximation $2^{\tc(G)/2}n^{O(1)}$-time algorithm where $\tc(G)$ is the twin cover number of an input graph $G$. 

\end{abstract}

\section{Introduction}\label{sec:intro}
Finding a dense subgraph is an important topic in graph mining.
There are many applications for real problems such as community detection in social networks \cite{Dourisboure2007},  \revised{spam detection \cite{Gibson2005}, and identification of molecular complexes in protein-protein interaction networks \cite{Bader2003}.}
In this paper, we study the \textsc{Densest $k$-Subgraph} problem, which is a graph optimization problem to find a dense structure in a graph.
More precisely, the problem is to find a vertex subset $S$ of size $k$ such that the number of edges in the subgraph induced by $S$ is maximized.

The \textsc{Densest $k$-Subgraph} problem is NP-hard due to the NP-hardness of \textsc{$k$-Clique}.
Thus, several papers study the parameterized complexity of \textsc{Densest $k$-Subgraph}.
In \cite{Bourgeois2017}, Bourgeois et al. show that \textsc{Densest $k$-Subgraph} can be solved in time $2^{\tw(G)}n^{O(1)}$  and exponential space, and in time $2^{\vc(G)}n^{O(1)}$ and polynomial space, respectively, where $\tw(G)$ and $\vc(G)$ are the tree-width and the vertex cover number of \revised{the} input graph $G$.
\revised{Here, if a problem has an algorithm with running time $f(p)n^{O(1)}$ where $f$ is some computable function, the problem is said to be {\em fixed-parameter tractable (FPT)} with respect to $p$.}
Therefore, the problem is FPT when parameterized by tree-width and vertex cover number, respectively.
Furthermore, it can be solved in time $2^{O(\cw(G) \log k)}n^{O(1)}$ where $\cw(G)$ is the clique-width of \revised{the} input graph  \cite{Broersma2013}.

On the other hand, \textsc{Densest $k$-Subgraph} is W[1]-hard when parameterized by $k$ due to the W[1]-completeness of \textsc{$k$-Clique}. Thus, the problem is unlikely to have any algorithm with running time $f(k)n^{O(1)}$.
Also, \textsc{Densest $k$-Subgraph} is W[1]-hard with respect to clique-width $\cw(G)$ and cannot be solved in time $2^{o(\cw(G) \log k)}n^{O(1)}$  unless Exponential Time Hypothesis (ETH) fails \cite{Broersma2013}.

In the \textsc{Densest $k$-Subgraph} problem, we seek a dense subgraph in a graph. 
Although the fixed-parameter algorithms with respect to tree-width or vertex cover number are already proposed, such parameters might be large if a graph has a dense subgraph.
In fact, if a graph contains a clique of size $k$, then the tree-width is at least $k-1$.
Thus, it is natural to consider structural parameters for graphs having dense sub-structures.



\subsection{Our contribution}
\begin{figure}[tbp]
    \centering
    \includegraphics[width=7cm]{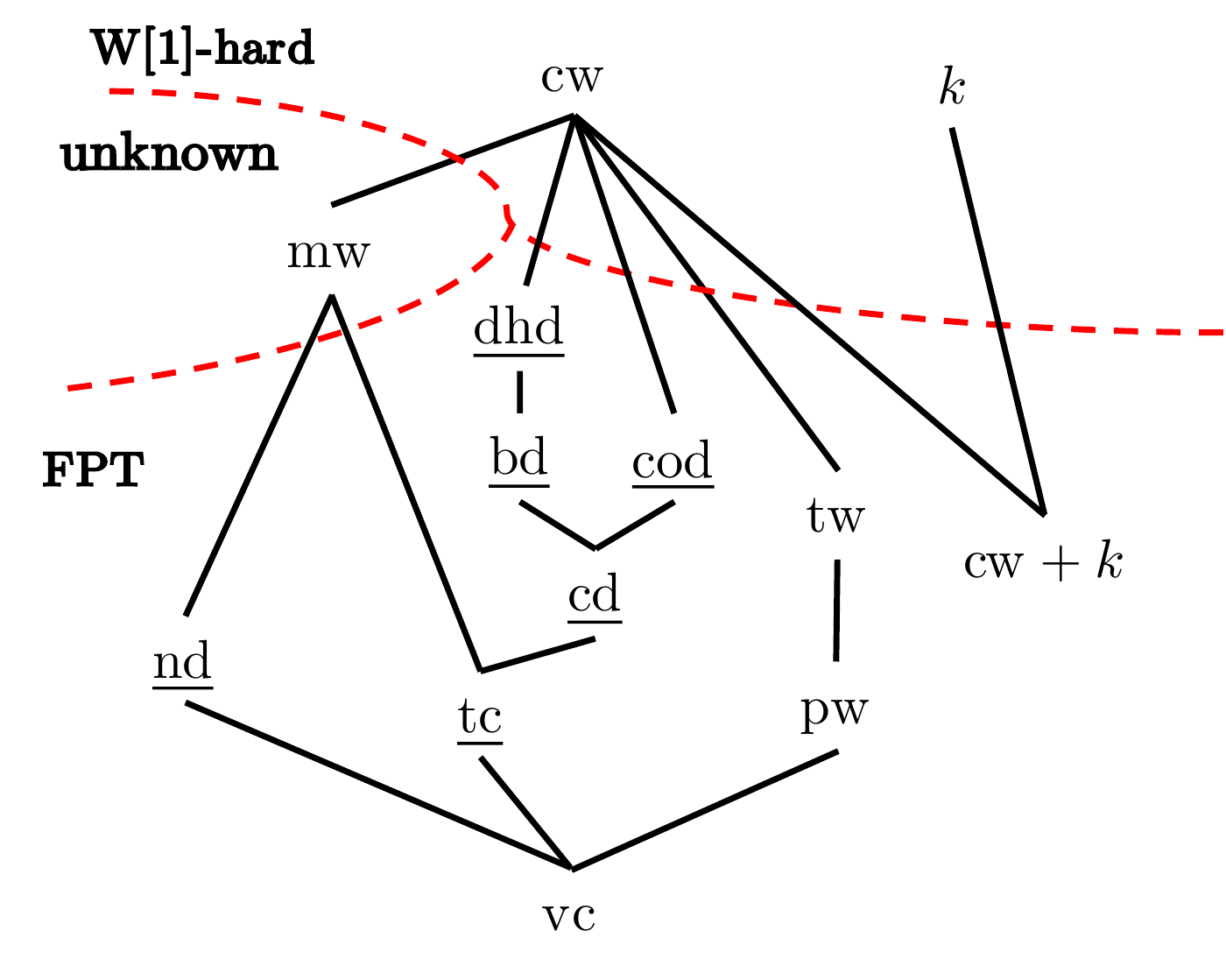}
    \caption{The relationship between graph parameters and the parameterized complexity of \textsc{Densest (Sparsest) $k$-Subgraph}. The parameters $\cw, \mw, \nd, \dhd, \db, \cod, \dc, \tc, \tw, \pw$, and $\vc$ mean clique-width, modular-width, neighborhood diversity, distance-hereditary deletion number, block deletion number, cograph deletion number, cluster deletion number, twin cover number, tree-width, path-width, and vertex cover number, respectively.
    Connections between two parameters imply that the above one is bounded by some function in the below one.
    The parameters with underlines are new results in this paper.}
    \label{fig:structural_parameter}
\end{figure}

In this paper, we study the fixed-parameter tractability of \textsc{Densest $k$-Subgraph} and \textsc{Sparsest $k$-Subgraph} by using structural parameters between clique-width and vertex cover number.
\textsc{Sparsest $k$-Subgraph} is the dual problem of \textsc{Densest $k$-Subgraph}.
The problem is to find a vertex subset $S$ of size $k$ such that the number of edges in the subgraph induced by $S$ is minimized. Clearly, \textsc{Sparsest $k$-Subgraph} is equivalent to \textsc{Densest $k$-Subgraph} on the complement of an input graph.
By using two structural parameters: {\em neighborhood diversity} and {\em block deletion number}, we show the following theorems.
\begin{theorem}\label{thm:DkS:nd}
\revised{\textsc{Densest $k$-Subgraph} and  \textsc{Sparsest $k$-Subgraph}  can be solved in time $f(\nd(G))n^{O(1)}$ where $\nd(G)$ is the neighborhood diversity of the input graph.}
\end{theorem}
\begin{theorem}\label{thm:DkS:db}
\revised{Given a block graph deletion set of size $\db(G)$, \textsc{Densest $k$-Subgraph} and \textsc{Sparsest $k$-Subgraph}  can be computed in time $O(2^{\db(G)}( (k^3+\db(G))n + m))$.}
\end{theorem}
\revised{Theorem \ref{thm:DkS:nd} and \ref{thm:DkS:db} imply that \textsc{Densest $k$-Subgraph} and \textsc{Sparsest $k$-Subgraph} are are fixed-parameter tractable when parameterized by the neighborhood diversity or the block deletion number.}

The neighborhood diversity is introduced by Lampis \cite{Lampis2012}. It is a structural parameter for {\em dense} graphs whereas the tree-width is for {\em sparse} graphs.
Also, the block deletion number of a graph is defined as the minimum number of vertices such that the graph becomes a block graph by removing them.
Both parameters are related to other structural parameters between clique-width and vertex cover number. 
Unlike tree-width, a graph having small neighborhood diversity or small block deletion number may have dense sub-structures.
For example, the neighborhood diversity and the block deletion number of a complete graph are 1, respectively, though the tree-width is $n-1$.
Figure \ref{fig:structural_parameter} shows the relationship between such graph parameters and the parameterized complexity of \textsc{Densest $k$-Subgraph} and \textsc{Sparsest $k$-Subgraph}.

Then, we generalize the approach of the block deletion number to the bounded clique-width deletion number.
\begin{theorem}\label{thm:bounded-cw}
Let $\mathcal{C}$ be a class of bounded clique-width graphs.
Then given a $\mathcal{C}$-deletion set $D$ of $G$, \textsc{Densest (Sparsest) $k$-Subgraph} can be solved in time $2^{|D|}n^{O(1)}$.
\end{theorem}

\revised{In particular, cographs and distance-hereditary graphs including block graphs are bounded clique-width graphs, respectively \cite{Courcelle2000,Golumbic2000}.
Therefore, Theorem \ref{thm:bounded-cw} implies that \textsc{Densest (Sparsest) $k$-Subgraph} is fixed-parameter tractable with respect to distance-hereditary deletion number and cograph deletion number, respectively.}


\textsc{Maximum $k$-Vertex Cover} is the problem to find a vertex subset $S$ of size $k$ that maximizes the number of edges such that at least one its endpoint is in $S$.
For any graph $G=(V,E)$, $S$ is an optimal solution of \textsc{Maximum $k$-Vertex Cover} in $G$ if and only if $V\setminus S$ is an optimal solution of \textsc{Sparsest $(n-k)$-Subgraph} in $G$ \cite{Bougeret2014}.
Thus, the following corollary holds.
  \begin{corollary}\label{col:MVC}
\revised{Let $\mathcal{C}$ be a class of bounded clique-width graphs. \textsc{Maximum $k$-Vertex Cover} can be solved in time $f(\nd(G))n^{O(1)}$, and $2^{|D|}n^{O(1)}$, respectively, where $\nd(G)$ is the neighborhood diversity and $D$ is a  $\mathcal{C}$-deletion set.}
  \end{corollary}

In this paper, we also give an FPT approximation algorithm for the \textsc{Densest $k$-Subgraph} problem.
 In \cite{Bourgeois2017}, Bourgeois et al. propose a $3$-approximation $2^{\vc(G)/2}n^{O(1)}$-time algorithm for \textsc{Densest $k$-Subgraph} where $\vc(G)$ is the vertex cover number of $G$.
In this paper, we improve the FPT approximation algorithm  in \cite{Bourgeois2017}.
Actually, we give a $2$-approximation $2^{\tc(G)/2}n^{O(1)}$-time algorithm parameterized by twin cover number $\tc(G)$. 
\begin{theorem}\label{thm:twincover:DkS}
There is a $2$-approximation  
$O(2^{\tc(G)/2} ( (k^3+\tc(G))n + m))$-time algorithm for \textsc{Densest $k$-Subgraph} where $\tc(G)$ is the twin cover number of an input graph $G$.
\end{theorem}
Note that for any graph $G$, $\tc(G)\le \vc(G)$ holds.

\subsection{Related work}
 \textsc{Densest $k$-Subgraph} is NP-hard due to the NP-hardness of $k$-clique.
It remains NP-hard on chordal graphs, comparability graphs, triangle-free graphs and bipartite graphs with maximum degree 3 \cite{Corneil1984,Feige1997}. On the other hand, it is solvable in polynomial time on cographs and split graphs \cite{Corneil1984}. 
For the parameterized complexity, \textsc{Densest $k$-Subgraph} is W[1]-hard when parameterized \revised{alone by $k$ or clique-width \cite{Broersma2013}.}
On the other hand, it is FPT when parameterized by tree-width and clique-width plus $k$, respectively \cite{Broersma2013,Bourgeois2017}.
On chordal graphs, \textsc{Densest $k$-Subgraph} parameterized by $k$ is FPT but it does not admit a polynomial kernel unless NP $\subseteq$ coNP/poly \cite{Bougeret2014}.


The approximability of \textsc{Densest $k$-Subgraph} is also well-studied. 
 Kortsarz and Peleg propose an $O(n^{0.3885})$-approximation algorithm \cite{Kortsarz1993}. 
 Then, Feige, Kortsarz, and  Peleg improve the approximation ratio to $O(n^{\delta})$ for some $\delta < 1/3$ \cite{Feige2001}.
 The best known approximation algorithm, proposed by  Bhaskara et al., archives the approximation ratio $O(n^{1/4+\epsilon})$ within $n^{O(1/\epsilon)}$ time, for any $\epsilon>0$ \cite{Bhaskara2010}.
For the hardness of approximation for \textsc{Densest $k$-Subgraph}, 
even the NP-hardness of approximation for a constant factor is still open.
To avoid this situation,  using  stronger assumptions, several inapproximability results are shown.
Khot  shows that there is no polynomial-time approximation scheme (PTAS) for \textsc{Densest $k$-Subgraph} assuming NP $\nsubseteq \cap_{\epsilon>0}$ BPTIME($2^{n^{\epsilon}}$) \cite{Khot2006}.
Moreover, Manurangsi proves that it does not admit an $n^{1/{(\log \log n)^c}}$-approximation algorithm for some constant $c>0$ under ETH \cite{Manurangsi2017}.
For the FPT inapproximability of \textsc{Densest $k$-Subgraph}, Chalermsook et al. show that there is no $k^{o(1)}$-FPT approximation algorithm assuming Gap-ETH \cite{Chalermsook2017}.
\revised{Manurangsi, Rubinstein, and Schramm show that there is no $o(k)$-approximation $f(k)n^{O(1)}$-time algorithm for \textsc{Densest $k$-Subgraph} under the Strongish Planted Clique Hypothesis (SPCH) \cite{MRS2021}.}

 \textsc{Sparsest $k$-Subgraph} is the \revised{complementary} problem of  \textsc{Densest $k$-Subgraph}. The problem is NP-hard because it is a generalization of \textsc{Independent Set}.
As with \textsc{Densest $k$-Subgraph}, \textsc{Sparsest $k$-Subgraph} is \revised{FPT} with respect to tree-width. 
Since for any graph $G$ and its complement $\bar{G}$ , $\cw(\bar{G})\le 2\cw(G)$ holds, \textsc{Sparsest $k$-Subgraph} is W[1]-hard when parameterized by clique-width and it cannot be solved in time $2^{o(\cw(G) \log k)}n^{O(1)}$ unless ETH fails, whereas it is solvable in time $2^{O(\cw(G) \log k)}n^{O(1)}$.

\textsc{Maximum $k$-Vertex Cover} is  NP-hard even on bipartite graphs \cite{Apollonio2014,Gwenael2015} and W[1]-hard when parameterized by $k$ \cite{Cai2008,Guo2007}.
Using the pipage rounding technique, Ageev and Sviridenko give a $0.75$-approximation algorithm \cite{Ageev1999}.
\revised{The best known approximation guarantee is $0.929$  \cite{RT2012:approx:maxvc,Manurangsi2019} and it is tight \cite{AS2019}.}
Although \textsc{Maximum $k$-Vertex Cover} is APX-hard \cite{Petrank1994}, there is a $(1-\epsilon)$-approximation algorithm that runs in time $(1/\epsilon)^k n^{O(1)}$ for any $\epsilon>0$ \cite{Manurangsi2019}.

\subsubsection*{Independent work}

\revised{Mizutani and Sullivan independently and simultaneously proved that \textsc{Densest $k$-Subgraph} is fixed-parameter tractable when parameterized by neighborhood diversity and cluster deletion number \cite{MS2022}.}


\section{Preliminaries}\label{sec:prelinimaries}
In this paper, we use the standard graph notations.
Let $G = (V(G), E(G))$ be an undirected graph, where $V(G)$ and $E(G)$ denote the
set of vertices and the set of edges, respectively. 
\revised{For simplicity, we sometimes denote $V$ and $E$ instead of $V(G)$ and $E(G)$, respectively.}
For a vertex subset $S$, we denote by $G[S]$ the subgraph induced by $S$.
We call a subgraph with $\ell$ vertices an {\em $\ell$-subgraph}.
For a vertex $v$, we denote by $N(v)$ the set of neighbors of $v$ and by $d(v)$ the degree of $v$.
We denote by $\cdot^T$ the transpose of a vector (resp., matrix).
For the basic definitions of parameterized complexity and structural parameters such as tree-width $\tw(G)$, we refer the reader to the book \cite{Cygan2015}.

\subsection{Integer quadratic programming}
In \textsc{Integer Quadratic Programming (IQP)}, the input consists of an $n\times n$ integer symmetric matrix ${\bf Q}$, an $m\times n$ integer matrix ${\bf A}$ and $m$-dimensional integer vector ${\bf b}$ and the task is to find an optimal solution ${\bf x}\in {\mathbb Z}^n$ to the following optimization problem.

\begin{align*}
    &\text{max (min)} \hspace{0.2cm}
    {\bf x}^{T}{\bf Q}{\bf x} +{\bf q}{\bf x}\\
 &\text{subject to}  \hspace{0.2cm} {\bf A}{\bf x}\le{\bf b},  {\bf C}{\bf x}= {\bf d}, {\bf x}\in {\mathbb Z}^n.
\end{align*}

Let $\alpha$ be the largest absolute value of an entry of ${\bf Q}$,  ${\bf A}$, ${\bf C}$, and ${\bf q}$. 
Then, Lokshtanov proved that \textsc{Integer Quadratic Programming} is fixed-parameter tractable when parameterized by $n$ and $\alpha$ \cite{Lokshtanov2015}.
\begin{theorem}[\cite{Lokshtanov2015}]\label{thm:IQP:fpt}
\textsc{Integer Quadratic Programming} can be solved in time  $f(n,\alpha)L^{O(1)}$ where $f$ is a computable function and $L$ is the length of the bit-representation of the input instance.
\end{theorem}


\subsection{Structural parameters}
In this subsection, we introduce the definitions of several structural parameters.

\subsubsection{Clique-width}

\revised{We first give the definition of clique-width (see also \cite{COURCELLE1993,courcelle_linear_2000}).
A vertex-labeled graph is a graph such that each vertex has an integer as a label. A \emph{$c$-graph} is a vertex-labeled graph with labels $\{1,\ldots,c\}$. We call a vertex labeled by $i$ an \emph{$i$-labeled vertex}. Then the \emph{clique-width} $\cw(G)$ of $G$ is defined as the minimum integer $c$ such that $G$ can be constructed by the following operations:
\begin{description}
    \item[{[O1]}] Create a vertex with label $i \in \{1, 2, \ldots, c\}$;
    \item[{[O2]}] Take a disjoint union of two $c$-graphs;
    \item[{[O3]}] For two labels $i$ and $j$, connect every pair of an $i$-labeled vertex and $j$-labeled vertex by an edge;
    \item[{[O4]}] Relabel all the labels of  $i$-labeled vertices to label $j$.
\end{description}}


For the above operations, we can construct a rooted binary tree, called a {\em $c$-expression tree}, such that each node corresponds to the above operations.
Nodes corresponding to (O1), (O2), (O3), and (O4) are called  {\em introduce} nodes, {\em union} nodes, {\em join} nodes, and {\em relabel} nodes, respectively.
The class of \emph{cographs} is equivalent to the class of graphs with clique-width at most 2 \cite{Courcelle2000}.
\color{black}

\subsubsection{Neighborhood diversity}

\begin{definition}[\cite{Lampis2012}]
Two vertices $u,v$ are called {\em twins} if both $u$ and $v$ have the same neighbors \revised{apart from $\{u, v\}$}. 
The {\em neighborhood diversity} $\nd(G)$  of a graph $G$ is the minimum
number such that $V$ can be partitioned into $\nd(G)$ sets of twin vertices.
\end{definition}
\revised{We call a set of twin vertices a {\em module}.
Note that every module forms either a clique or an independent set and
two modules either are completely joined by edges or have no edge between them in $G$. 
Given a graph, its neighborhood diversity and the corresponding partition can be computed in linear time \cite{Lampis2012,McConnell1999,Tedder2008}.}

\subsubsection{$\mathcal{G}$-deletion number}
For a graph class $\mathcal{G}$, a \emph{$\mathcal{G}$-deletion set} $D\subseteq V$ of $G$ is \revised{a} set of vertices such that $G[V\setminus D]$ is in $\mathcal{G}$.
The \emph{$\mathcal{G}$-deletion number} $\dd_{\mathcal G}(G)$ of a graph $G$ is defined as the minimum size of \revised{a} $\mathcal{G}$-deletion set in $G$.







\subsubsection{Twin cover and vertex cover}

\revised{A set of vertices $X$ is \revised{a} {\em twin-cover} of $G$ if any edge $\{u,v\}$ satisfies either
(i) $u\in X$ or $v\in X$, or (ii) $N[u]=N[v]$.}
A {\em vertex cover} $X$ is \revised{a} set of vertices such that for every edge, at least one endpoint is in $X$.

\subsubsection{Relationship between parameters}
\revised{For the clique-width $\cw(G)$, the tree-width $\tw(G)$, the modular-width $\mw(G)$, the neighborhood diversity $\nd(G)$, the block deletion number $\db(G)$, the cluster deletion number $\dc(G)$, the twin cover number $\tc(G)$, and  the vertex cover number $\vc(G)$, the following relationship holds.
\begin{proposition}[\cite{Bodlaender1995b,Courcelle2000,Gajarsky2013,Ganian2015,Lampis2012}]\label{prop:parameters}
For any graph $G$, the following inequalities hold: $\cw(G)\le 2^{\tw(G)+1}+1$, $\tw(G)\le \vc(G)$, $\cw(G)\le \mw(G)+2$, $\mw(G)\le \nd(G)  \le 2^{\tc(G)}+\tc(G)$,  and $\db(G)\le \dc(G)\le \tc(G)\le \vc(G)$.
\end{proposition}
}


\begin{proposition}\label{prop:clique-width:parameters}
Let $\mathcal{C}$ be a class of graphs of clique-width at most $c$ and $\dd_{\mathcal{C}}(G)$ be the $\mathcal{C}$-deletion number of a graph $G$.
Then for any graph $G$, $\cw(G)\le 2^{\dd_{\mathcal{C}}(G)+c+1}$ holds.
\end{proposition}
\begin{proof}
We generalize the proof of Lemma 4.17 in \cite{Sorge2013}.
Let $D$ be a minimum $\mathcal{C}$-deletion set of $G$.
Then the clique-width of $G[V\setminus D]$ is at most $c$.
Since the rank-width of a graph is at most the clique-width \cite{Oum2006}, the rank-width $\rankw(G\setminus D)$ of $G[V\setminus D]$ is at most $c$.
It is known that deleting a vertex decreases the rank-width by at most $1$ \cite{Hlineny2007}, and hence the rank-width $\rw(G)$ of $G$ is at most $c+\dd_{\mathcal{C}}(G)$.
Because  $\cw(G)\le 2^{\rankw(G)+1}$ holds \cite{Oum2006}, we have $\cw(G)\le 2^{\dd_{\mathcal{C}}(G)+c+1}$.
\end{proof}

\begin{corollary}
For any graph $G$, $\cw(G)\le 2^{\db(G)+4} \le 2^{\dhd(G)+4}$ holds \revised{where $\db(G)$ is the block deletion number and $\dhd(G)$ is the distance-hereditary  deletion number.}
\end{corollary}
\begin{proof}
The clique-width of a distance-hereditary graph is at most 3 \cite{Golumbic2000}.
\end{proof}
\section{Neighborhood Diversity}\label{sec:nd}
In this section, we prove Theorem \ref{thm:DkS:nd}.
\revised{Let ${\mathcal M}=\{M_1, \ldots, M_{\nd(G)}\}$ be a set of modules with neighborhood diversity $\nd(G)$.  The {\em type graph} $Q=(\mathcal{M}, E_Q)$ of $G$ is a graph such that each vertex is a module of $G$ and there is an edge between modules if and only if two modules are completely joined in $G$. }
We denote by $S$ the solution set.
Then we define a variable $x_i$ for each module $M_i$.
A variable $x_i$ represents the number of vertices in $S\cap M_i$.
Moreover, we define $m_i=|M_i|$ for  $M_i$.
Without loss of generality, we assume that $\mathcal{M}_C=\{M_1, \ldots, M_p\}$ is the set of modules that form cliques and $\mathcal{M}_I=\{M_{p+1}, \ldots, M_{\nd(G)}\}$ is the set of modules that are independent sets.

Then \textsc{Densest $k$-Subgraph} can be formulated as the following integer quadratic programming problem.
\begin{mpproblem}{IQP-D$k$S}
\label{mpprob:P}
\begin{align*}
\text{maximize} &\sum_{\{M_i,M_j\}\in E_Q}x_ix_j + \sum_{i\in \{1,\ldots,p\}}x_i(x_i-1)/2   \\
 \text{subject to}  &\sum_{i\in\{1,\ldots,{\nd(G)}\}}x_i=k \\
 &0\le  x_i\le m_i \ \ \ \ \ \ \ \ \ \  \forall i\in\{1,\ldots,{\nd(G)}\} \\
& x_i\in {\mathbb Z} \ \ \ \ \ \ \ \  \ \ \ \ \ \ \ \ \ \ \forall i\in\{1,\ldots,{\nd(G)}\} 
\end{align*}
\end{mpproblem}
In the objective function of (IQP-D$k$S), the first term represents the sum of edges between modules in $G[S]$. The second term represents the sum of edges inside modules.
Also, the first constraint represents the size constraint.

Without loss of generality, we replace the objective function by: $$\sum_{ \{M_i,M_j\}\in E_Q}2x_ix_j + \sum_{i\in \{1,\ldots,p\}}x_i(x_i-1).$$




It is easily seen that (IQP-D$k$S) can be represented as \textsc{Integer Quadratic Programming} such that the  largest  absolute  value  of  entries  of matrices in  (IQP-D$k$S) is 1.
The number of variables in (IQP-D$k$S) is ${\nd(G)}$ and the number of constraints is $2{\nd(G)}+1$. 
By applying Theorem \ref{thm:IQP:fpt}, we complete the proof of Theorem \ref{thm:DkS:nd}.
Furthermore, the minimization version of (IQP-D$k$S) is the IQP formulation of 
\textsc{Sparsest $k$-Subgraph}. 
Thus, \textsc{Sparsest $k$-Subgraph} is also fixed-parameter tractable when parameterized by neighborhood diversity.  

\section{Block Deletion Number}\label{sec:block}
In this section, we prove Theorem \ref{thm:DkS:db}. 
\revised{To show this, we define \textsc{Densest (Sparsest) $k$-Subgraph with Weighted Vertices} as an 
auxiliary problem.
\medskip
\begin{breakbox}
\noindent \textsc{Densest (Sparsest) $k$-Subgraph with Weighted Vertices}\\
\noindent {\bf Input:} A graph $G=(V,E)$ with vertex weight $w_v$ and an integer $k$.\\
\noindent {\bf Output:} A vertex set $V'$ of size $k$ that maximizes (minimizes) $|\{\{u,v\}\in E\mid u,v\in V'\}|+\sum_{v\in V'}w_v$.
\end{breakbox}
\medskip
Note that if $w_v=0$ for every $v\in V$, \textsc{Densest (Sparsest) $k$-Subgraph with Weighted Vertices} is equivalent to \textsc{Densest (Sparsest) $k$-Subgraph}.
}

\begin{lemma}\label{lem:Deletion_Set}
Let $\mathcal{G}$ be a graph class and $D$ be a $\mathcal{G}$-deletion set.
 If \textsc{Densest (Sparsest) $k$-Subgraph with Weighted Vertices} can be solved in time $T$, then \textsc{Densest (Sparsest) $k$-Subgraph} is solvable in time $O(2^{|D|}(|D|n + T))$ when $D$ is given. In particular, if $T=(nw_{\max})^{O(1)}$, it can be solved in time $2^{|D|}n^{O(1)}$.
\end{lemma}
\begin{proof}
Given a $\mathcal{G}$-deletion set $D$ of $G$,
we first guess $2^{|D|}$ candidates of partial solutions in $G[D]$.
For each candidate $S\subseteq D$, we define the weight $w_v$ of $v\in V\setminus D$ as the number of neighbors in $D$.
The weight $w_v$ for $v\in V\setminus D$ means the number of additional edges of the solution when $v$ is added to the solution $S$.
Every weight can be computed in time $O(|D| n)$.
Let $k'=k-|S\cap D|$.
Then we solve \textsc{Densest (Sparsest) $k'$-Subgraph with Weighted Vertices} with respect to $w_v$ in $G[V\setminus D]$.
Thus, the total running time is $O(2^{|D|}(|D|n + T))$.

In \textsc{Densest (Sparsest) $k'$-Subgraph with Weighted Vertices} with respect to $w_v$ in $G[V\setminus D]$, $w_{\max}=\max_{v\in {V\setminus D}}w_v$ is bounded by $|D|\le n$, and hence it can be solved in time $T=(nw_{\max})^{O(1)}=n^{O(1)}$.
Therefore, the total running time is $2^{|D|}n^{O(1)}$.
\end{proof}

\color{black}

We then show that \textsc{Densest $k$-Subgraph with Weighted Vertices}  can be solved in polynomial time  on block graphs.

\revised{
\begin{figure}[tbp]
    \centering
    \includegraphics[width=12cm]{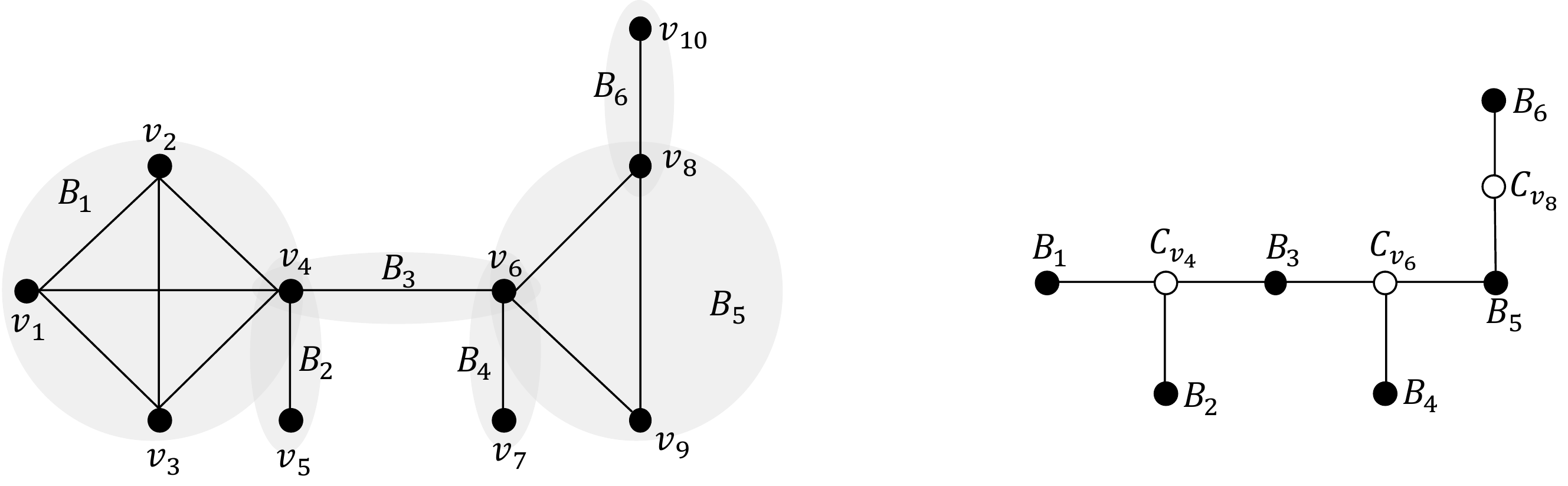}
    \caption{An example of a block graph (left) and its block cut tree (right). In the block-cut tree, the black nodes represent block nodes and the white nodes represent cut nodes.}
    \label{fig:block-cut-tree}
\end{figure}
Suppose that a block graph $G$ is connected.
A vertex $v$ is called a {\em cut vertex} if  $G[V\setminus \{v\}]$ has at least two components.
By the definition of a block graph, two blocks share exactly one cut vertex.
Let ${\mathcal B}=\{B_1,\ldots, B_{\beta}\}$  be the set of blocks where $\beta$ is the number of 2-connected components.
For each cut vertex $v$, we define $C_v=\{v\}$. Let ${\mathcal C}=\{C_{v_1},\ldots, C_{v_\gamma}\}$ where $\gamma$ is the number of cut vertices and let ${\mathcal X}={\mathcal B}\cup {\mathcal C}$.
Then the {\em block-cut tree} of $G$ is a tree ${\mathcal T}=({\mathcal X}, {\mathcal E})$ where $\mathcal{E}=\{\{B_i,C_v\}\mid B_i\cap C_v=\{v\}, B_i\in \mathcal{B}, C_v\in \mathcal {C}\}$ (see Figure \ref{fig:block-cut-tree}).
In the block-cut tree, $B_i$ is called a {\em block node} and $C_v$ is called a {\em cut node}.
We denote by $B_r$ the {\em root} node in ${\mathcal T}$.
The block-cut tree of $G$ can be computed in linear time \cite{Tarjan1972}.
}

\begin{lemma}\label{lem:DkSWV:block}
\textsc{Densest $k$-Subgraph with Weighted Vertices} on block graphs can be solved in time $O(k^3n+m)$.
\end{lemma}
\begin{proof}
Suppose that a block graph $G$ is connected.
For each node $i$ in a block-cut tree, we define the set of vertices $V_i\subseteq V$ as the union of all nodes $B_j$ such that $j=i$ (if node $i$ is a block node) or $j$ is a descendant of $i$.
Note that $V=V_r$.

Our algorithm is based on dynamic programming on the block-cut tree $\TT=(\XX,\EE)$ of $G$.
Except for the root node,  we suppose that a block node has its parent cut node $C_v=\{v\}$ and we call $v\in C_v$ the {\em parent cut vertex}. Note that the parent node of a block node is a cut node.
For $p\in \{0,1\}$ and $\ell\in \{0, \ldots, k\}$, let $A[i,p,\ell]$ be the maximum value  of $\ell$-subgraphs in $G[V_i]$.
The value $p$ indicates whether the parent cut vertex is contained in the solution or not.
If $p=0$, the parent cut vertex  is not contained in the solution, and otherwise, it is contained.
In the root node $B_r$, which is a block node,  $A[r,\ell]$ is defined as the maximum value  of $\ell$-subgraphs on a block graph.
Thus, $A[r,k]$ is the optimal value of \textsc{Densest $k$-Subgraph with Weighted Vertices}.  
Here, we define $A[i,p,\ell]=-\infty$ as an invalid case.

In the following, we define the recursive formulas for each node on a block cut tree.

\medskip
\noindent{\bf Leaf block node: }
In a leaf block node $B_i$, let $v\in B_i$ be the parent cut vertex.
Without loss of generality, we suppose \revised{$B_i\setminus \{v\}=\{u_1, u_2,\ldots, u_{|B_i|-1}\}$ where $w_{u_1}\ge w_{u_2}\ge \ldots \ge w_{u_{|B_i|-1}}$.}
For $\ell\ge 1$,  we define $a[\ell]=\sum_{i=1}^{\ell} w_{u_i}$, which is the sum of weights of the  top $\ell$ vertices in the order of weights in $B_i\setminus \{v\}$. For $\ell=0$,  we set $a[0]=0$.
Then $ A[i,p,\ell]$ can be defined as follows: 
\begin{align*}
 A[i,p,\ell]=\begin{cases} \ell(\ell-1)/2  +  a[\ell] & \mbox{if $p=0$ and $0\le \ell\le |B_i|-1$}\\
 \ell(\ell-1)/2 + a[\ell-1] + w_v & \mbox{if $p=1$ and $1\le \ell\le |B_i|-1$}\\
-\infty& \mbox{otherwise}
 \end{cases}.
\end{align*}
The first case is when the parent cut vertex $v$ is not contained in the solution, and the second case is when $v$ is contained in the solution.

\medskip
\noindent{\bf Internal cut node: }
In an internal cut node $C_{v_i}=\{v_i\}$, we compute  the maximum value $A[v_i,p,\ell]$ of  $\ell$-subgraphs  in $G[V_{v_i}]$.
If $p=1$, the solution contains the cut vertex $v_i$, and otherwise it does not.
Suppose that the internal cut node $C_{v_i}$ has $t$ subtrees in $\mathcal{T}$.
Let $T_1, \ldots, T_t$ be the subtrees of $C_{v_i}$ whose root nodes are $C_{v_i}$'s children $B_{j_1}, \ldots, B_{j_t}$.

Using dynamic programming, we compute $A[v_i,p,\ell]$. 
For $p\in \{0,1\}$, $s\in \{1, \ldots, t\}$, and $\ell\in \{0,\ldots, k\}$, we compute $c_{v_i}[s,p,\ell]$ that represents the maximum value of  $\ell$-subgraphs  in the subgraph induced by vertices in the subtrees  $T_1, \ldots, T_s$.
Note that $A[v_i,p,\ell]=c_{v_i}[t,p,\ell]$.
At the base case, we set $c_{v_i}[1,p,\ell]=A[j_1,p,\ell]$ for each $p$ and $\ell$.
Then $c_{v_i}[s,p,\ell]$ can be recursively computed by using the following formula:
\begin{align*}
c_{v_i}[s,p,\ell]=\begin{cases} \max_{\ell'+\ell''=\ell}\{ A[j_s,p,\ell']+c_{v_i}[s-1,p,\ell'']\} & \mbox{if $p=0$}\\
\max_{\ell'+\ell''=\ell+1}\{ A[j_s,p,\ell']+c_{v_i}[s-1,p,\ell'']\} & \mbox{if $p=1$}
\end{cases}.
\end{align*}
In the case of $p=1$, we have to consider the double counting of $v$, and hence we set $\ell'+\ell''=\ell+1$.
Finally, we set $A[v_i,p,\ell]=c_{v_i}[t,p,\ell]$

\medskip
\noindent{\bf Internal block node: }
In an internal block node $B_i$, let $B_i\setminus \{v\}=\{u_1, \ldots, u_{|B_i|-1}\}$.
Here, for $u\in B_i\setminus \{v\}$ we set:
\begin{align*}
w'_u[p,\ell]=\begin{cases} A[u,p,\ell] & \mbox{if  $u$ is a cut vertex}\\
w_u & \mbox{if $p=1$, $\ell=1$, and $u$ is not a cut vertex}\\
0 & \mbox{if $p=0$ and $\ell=0$, and $u$ is not a cut vertex}\\
-\infty & \mbox{otherwise}
\end{cases}.
\end{align*}
In a sense,  \revised{$w'_u[p,\ell]$} is the new weight of $u$ taking into account the maximum value of  $\ell$-subgraphs in the subtree having $C_u=\{u\}$ as the root.

For $\alpha\in \{0,\ldots, \ell\}$, we  compute $a[j,\alpha,\ell]$, which represents the maximum value of $\ell$-subgraphs that contains $\alpha$ vertices among $\{u_1,\ldots, u_j\}$ and $\ell-\alpha$ vertices in the subtrees of $B_i$.
Here, for any $j$ and $\ell$, we set $a[j,\alpha,\ell]=-\infty$ if $\alpha=-1$. 
For $j=1$, we set:
\begin{align*}
a[j,\alpha,\ell]=\begin{cases}w'_{u_1}[0,\ell] & \mbox{if $\alpha=0$}\\
w'_{u_1}[1,\ell]  & \mbox{if $\alpha=1$}\\
-\infty & \mbox{otherwise}\\
\end{cases}.
\end{align*}
The first case is when $u_1$ is not contained in the solution and the second case is when $u_1$ is in the solution.

For $j\in \{2, \ldots, |B_i|-1\}$ and $0\le \alpha\le \ell$, we compute $a[j,\alpha,\ell]$ by using the following recursive formula:
\begin{align*}
a[j,\alpha,\ell]= \max_{\ell'+\ell''=\ell}\max \left\{ 
\begin{array}{ll}
a [j-1,\alpha,\ell']+w'_{u_{j}}[0,\ell''],\\ a[j-1,\alpha-1,\ell']+w'_{u_{j}}[1, \ell'']
\end{array}\right\}.
\end{align*}
The first case is when $u_j$ is not contained in the solution and the second case is when $u_j$ is in the solution.
After the calculation, we obtain $a[{|B_i|-1},\alpha,\ell]$ for every $\alpha$ and $\ell$.
Let $a[\alpha,\ell]=a[|B_i|-1,\alpha,\ell]$.
Finally, the maximum value  of $\ell$-subgraphs on a node $B_i$ can be computed as follows:
\revised{\begin{align*}
A[i,p,\ell]=\begin{cases}\max_{0\le \alpha \le \ell} \left\{a[\alpha,\ell] + \alpha(\alpha-1)/2 \right\}  & \mbox{if $p=0$}\\
\max_{0\le \alpha \le\ell-1} \left\{a[\alpha,\ell-1] + \alpha(\alpha+1)/2 \right\} + w_v & \mbox{else}
\end{cases}.
\end{align*}}
The first case is when the parent cut vertex $v$ is not contained in the solution and the second case is when $v$ is in the solution.
Note that $G[B_i]$ forms a clique.

\medskip
\noindent{\bf Root node: }
In the root node $X_r$, we compute $a[r,\alpha,\ell]$ that represents the maximum value of $\ell$-subgraphs that contains $\alpha$ vertices in $B_r$ and $\ell-\alpha$ vertices in the subtrees of $B_r$ by the same way as internal block nodes.
Then, we set $A[r,\ell] = \max_{0\le \alpha \le \ell }  \{a[r,\alpha,\ell] + \alpha(\alpha-1)/2\}$ as the maximum value of $\ell$-subgraph in $G=G[V_r]$.
Finally, we obtain the maximum value  $A[r,k]$ of  $k$-subgraphs on $G$. 

\bigskip
In the algorithm, we first compute the block-cut tree of $G$ in linear time \cite{Tarjan1972}.
Since there are at most $n$ nodes in a block graph, the number of $A[i,p,\ell]$'s is $O(kn)$.
For $a[j,\alpha,\ell]$, each $j$ corresponds to a vertex in a node and it appears exactly once.
Thus, the number of  $a[j,\alpha,\ell]$'s is  $O(k^2n)$.
For each $u\in B_i\setminus \{v\}$, $w'_u[p,\ell]$ is defined.  Thus, the number of $w'_u[p,\ell]$'s  is $O(kn)$.
Finally,   the number of children nodes for all $C_v$'s is at most the number of block nodes. Therefore, the total number of $c_v[j,p,\ell]$'s for all $C_v$'s is bounded by $O(kn)$.

Since for fixed $p$, $\ell$, and $\alpha$, $A[i,p,\ell]$, $a[j,\alpha,\ell]$, $c_u[j,p,\ell]$, and $w'_u[p,\ell]$ can be computed in time $O(k)$,
we can compute $A[r,k]$ in time $O(k^3n)$.

It is easily seen that the algorithm can be applied to the case that $G$ is not connected.
We first compute the optimal value of \textsc{Densest $\ell$-Subgraph with Weighted Vertices} for every $\ell$ on each connected component of $G$.
Then we again use dynamic programming.
Let $G_1, \ldots, G_t$ be the connected components of $G$.
For $i\in \{1, \ldots, t\}$, we denote by $A_i[r,\ell]$ the optimal value of \textsc{Densest $\ell$-Subgraph with Weighted Vertices} on $G_i$.
Let $B[i,\ell]$ be the maximum value of $\ell$-subgraphs on the union of $G_1, \ldots, G_i$.
At the base case, we set $B[1,\ell] = A_1[r,\ell]$ for each $0\le \ell\le k$.
Then for each $i\in \{1, \ldots, t\}$ and $\ell\in \{0,\ldots,k\}$, we set $B[i,\ell] = \max_{0\le \ell' \le \ell} \{B[i-1,\ell'] + A_i[r,\ell-\ell']\}$.
Clearly, $B[t,\ell]$ is the maximum value of $\ell$-subgraphs on $G$ and it can be computed in time $O(kt)$. 
Since $t\le n$, the total running time is $O(k^3n+m)$.
\end{proof}

By using Lemma \ref{lem:DkSWV:block}, we complete the proof of Theorem \ref{thm:DkS:db}.

Since \textsc{Block Vertex Deletion} and \textsc{Cluster Vertex Deletion} can be computed in time $4^{\db(G)}n^{O(1)}$ \cite{Agrawal2016} and $O(1.9102^{\dc(G)}\dc(G)(n + m))$ \cite{Boral2016}, respectively, we also obtain the following corollary.
\begin{corollary}
\textsc{Densest $k$-Subgraph} and \textsc{Sparsest $k$-Subgraph} can be solved in time   $4^{\db(G)}n^{O(1)}$ and $O(2^{\dc(G)} \dc(G) (k^3n +m))$ where $\db(G)$ and $\dc(G)$ are the block deletion number and the cluster deletion number of an input graph, respectively.
\end{corollary}

\section{Generalization to Bounded Clique-width Deletion Number}

In this section, we generalize the results in Section \ref{sec:block}.
Let $\mathcal{C}$ be a class of bounded clique-width graphs.
Because the clique-width of a block graph is at most 3, $\mathcal{C}$ includes block graphs.

We show that \textsc{Densest (Sparsest) $k$-Subgraph with Weighted Vertices} can be solved in time  \revised{$k^{O(c)}n$}  on bounded clique-width graphs.
\begin{lemma}\label{lem:cw}
Given a $c$-expression tree $\TT$ with $O(n)$ nodes, \textsc{Densest (Sparsest) $k$-Subgraph with Weighted Vertices} can be solved in time \revised{$O((k+1)^{2(c+1)}n)$}.
\end{lemma}
\begin{proof}
We give an algorithm based on dynamic programming  by modifying the algorithm proposed in Theorem 1 in \cite{Broersma2013}.

For a node $t$ in $\TT$, we denote by $G_t=(V_t, E_t)$ the vertex-labeled subgraph of $G$ represented by node $t$.
Then, we define a DP table for each node $t$ which stores $c+1$ positive integers  $\ell$ and $s_1, \ldots, s_c$.

\revised{For each row in a DP table of node $t$, we define $D_t[\ell, s_1,s_2,\ldots,s_c]$ as the sum of the weights of vertices and the number of edges in $G[S]$ induced by $S\subseteq V_t$ such that:  
\begin{itemize}
\item  for every label $i$, $s_i$ is the number of $i$-labeled vertices in $S$
\item $|S|=\sum_{i\in \{1,\ldots, c\}}s_i = \ell$.
\end{itemize}
}
For invalid cases, we set $D_t[\ell, s_1,s_2,\ldots,s_c]=-\infty$.
Note that $0\le \ell\le k$ and $0\le s_i\le k$ for each $i$.
The size of a DP table in a node is at most \revised{$ (k+1)^{c+1}$}.
In the root node $r$ of a $c$-expression tree, a row maximizing $D_t[k,  s_1,s_2,\ldots,s_c]$ in the DP table indicates an optimal solution.

\medskip
\revised{\noindent{\bf Introduce node: }
In an introduce node $t$,  suppose that an $i$-labeled vertex $v$ is introduced.
We define the recursive formula as follows.
\begin{align*}
    D_t[\ell, s_1,s_2,\ldots,s_c] = \begin{cases} 
    w_v & \mbox{if $\ell = 1$, $s_i=1$, and $s_l=0$ for $l\neq i$}\\
    0 & \mbox{if  $\ell =0$ and $s_l=0$ for every $l\in \{1,\ldots, c\}$}\\
    -\infty & \mbox{otherwise}
    \end{cases}.
\end{align*}
}
In the first case, $v$ is contained in $S$, and in the second case $v$ is not contained in $S$.
Note that $G_t$ consists of exactly one $i$-vertex $v$.

\medskip
\noindent{\bf Union node: }
In a union node $t$, let $t_1$ and $t_2$  be its children nodes. 
The graph $G_t$ consists of the disjoint union of $G_{t_1}$ and $G_{t_2}$.
\revised{Thus the densest $\ell$-subgraph of $G_t$ is obtained from the disjoint union of the densest $\ell_1$-subgraph in $G_{t_1}$ and the densest $\ell_2$-subgraph in $G_{t_2}$ such that  $\ell_1 + \ell_2 = \ell$.
Therefore, we define the recursive formula as follows:
\begin{align*}
    D_t[\ell, s_1,s_2,\ldots,s_c]= \max_{\substack{\ell_1+\ell_2=\ell\\ s^1_l+s^2_l=s_l, \forall l\in \{1, \ldots, c\}}} \left(D_{t_1}[\ell_1, s^1_1,s^1_2,\ldots,s^1_c] + D_{t_2}[\ell_2, s^2_1,s^2_2,\ldots,s^2_c]\right).
\end{align*}
}

\medskip
\noindent{\bf Join node: }
\revised{In a join node $t$, let $t'$ be its child node.
Because all $i$-labeled vertices and all $j$-labeled vertices are joined by adding all possible edges between them, we define $D_t[\ell, s_1,s_2,\ldots,s_c] =  D_{t'}[\ell, s_1,s_2,\ldots,s_c] + s_is_j$.
}

\medskip
\noindent{\bf Relabel node: }
In a relabel node $t$, label $i$ is changed to label $j$. After the relabeling, $s_i = 0$ and $s_j = s'_i + s'_j$.
\revised{Thus, if $s_i=0$, $s'_i+s'_j=s_j$, and $s'_l=s_l$ for every $l\in \{1, \ldots, c\}\setminus \{i,j\}$, we set $D_t[\ell, s_1,s_2,\ldots,s_c] =D_{t'}[\ell, s'_1,s'_2,\ldots,s'_c]$, and otherwise $D_t[\ell, s_1,s_2,\ldots,s_c] =  -\infty $.
\begin{align*}
    D_t[\ell, s_1,s_2,\ldots,s_c]=\begin{cases} \max_{\substack{s'_i+s'_j=s_j\\ s'_l=s_l, \forall l\in \{1, \ldots, c\}\setminus \{i,j\}}} D_{t'}[\ell, s'_1,s'_2,\ldots,s'_c] & \mbox{if $s_i=0$} \\
    -\infty & \mbox{otherwise}
    \end{cases}.
\end{align*}
}

\bigskip
It is easily seen that each DP table can be computed in polynomial time. 
Because the size of a DP table is at most \revised{$(k+1)^{c+1}$}, the total running time is \revised{$O((k+1)^{2(c+1)}n)$}.
\end{proof}

One can compute a $(2^{\cw(G)+1}-1)$-expression tree with $O(n)$ nodes of a graph of clique-width $\cw(G)$ in time $O(n^3)$ \cite{Hlineny2008,Oum2008,Oum2006}.
Thus, \textsc{Densest (Sparsest) $k$-Subgraph with Weighted Vertices} on bounded clique-width graphs can be computed in time \revised{$n^{O(1)}$}.
Therefore, we immediately obtain Theorem \ref{thm:bounded-cw} by Lemma \ref{lem:Deletion_Set}.

\section{FPT Approximation Algorithm}\label{sec:parameter_approx}

In this section, we prove Theorem \ref{thm:twincover:DkS}.
Actually, we give a $2$-approximation $2^{\db(G)/2}n^{O(1)}$-time algorithm for \textsc{Densest $k$-Subgraph}.
\begin{theorem}\label{thm:approx:DkS}
Given a block deletion set of size $\db(G)$, there is a $2$-approximation algorithm for \textsc{Densest $k$-Subgraph} that runs in time  $O(2^{\db(G)/2} ( (k^3+\db(G))n + m))$.
\end{theorem}
\begin{proof}
\begin{figure}[tbp]
    \centering
    \includegraphics[width=7cm]{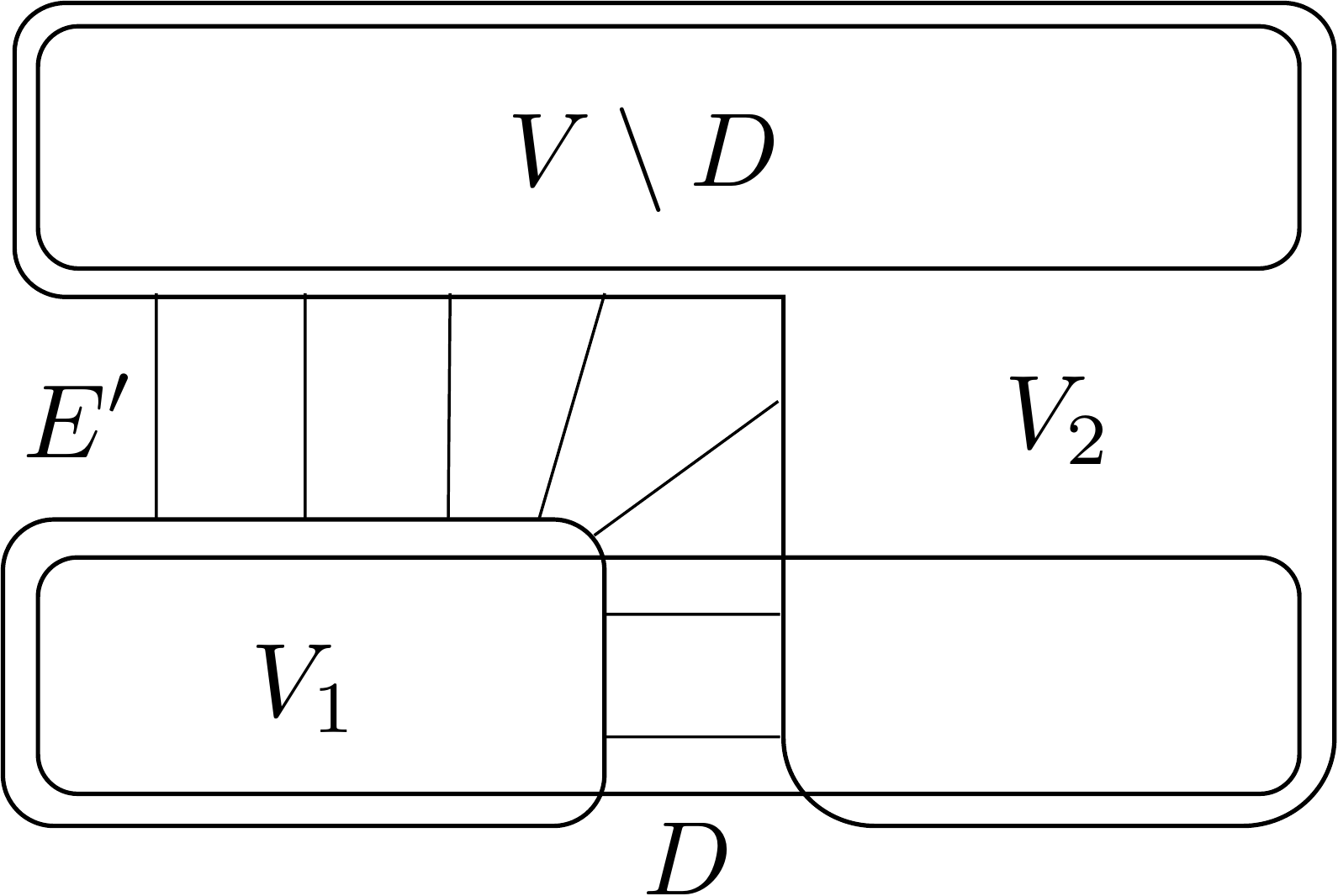}
    \caption{Partition $(V_1, V_2)$ of $V$.}
    \label{fig:partition}
\end{figure}
Our algorithm is based on an FPT approximation algorithm  proposed in \cite{Bourgeois2017}.
\revised{Given a block deletion set  $D$ of size $\db(G)$, let $V_1\subseteq D$ be a subset of $D$ of size $\lfloor \db(G)/2\rfloor$ and let $V_2 = V\setminus V_1$ (see Figure \ref{fig:partition}). Note that $|V_2\cap D|=\lceil \db(G)/2\rceil$}.
Let $E_1$ and $E_2$ be the set of edges in $G[V_1]$ and $G[V_2]$, respectively.
Also, let $E'=E\setminus (E_1\cup E_2)$.

\revised{We solve \textsc{Densest $k$-Subgraph} for two graphs $G'=(V_1\cup V_2, E')$ and $G''=(V_1\cup V_2, E_1\cup E_2)$. Note that $E'$ and $E_1\cup E_2$ are disjoint and $G'$ is bipartite. 
First, we solve \textsc{Densest $k$-Subgraph} in $G''$.  To do this, we solve \textsc{Densest $i$-Subgraph} for every $i\in \{0, \ldots, k\}$ in $G[V_1]$ and $G[V_2]$, respectively.}
Let $S_j^i$ and $s^i_j$ be the optimal solution and its value of \textsc{Densest $i$-Subgraph} in $G[V_j]$ for $j\in \{1,2\}$.
We can observe that an optimal solution in $G''=(V_1\cup V_2, E_1\cup E_2)$ is $S^{i_1}_1\cup S^{i_2}_2$ such that $i_1+i_2=k$ and $s_1^{i_1}+s_2^{i_2}$ is maximized  because 
$G''$ consists of two disjoint graphs $G[V_1]$ and $G[V_2]$.
Also, we solve \textsc{Densest $k$-Subgraph} on $G'$.
Let $S'$ and $S''$ be optimal solutions in $G'$ and $G''$, respectively.
Finally, we output the larger of $S'$ and $S''$.

Since one of $G[S']$ and $G[S'']$  has at least half of the optimal number of edges, the approximation ratio of this algorithm is $2$.

Finally, we estimate the running time.
Because \revised{$|V_1|=\lfloor \db(G)/2\rfloor$} and the block deletion number of  $G[V_2]$ is \revised{$\lceil \db(G)/2\rceil$}, we can compute \textsc{Densest $k$-Subgraph} in $G[V_1]$ and $G[V_2]$ in time $O(2^{\db(G)/2}((k^3+\db(G))n + m))$, respectively.
Thus, $S''$ can be computed in time $O(2^{\db(G)/2} ( (k^3+\db(G))n + m))+O(k)=O(2^{\db(G)/2} ( (k^3+\db(G))n + m))$.
Furthermore, since  $G'$ is a bipartite graph for $V_1$ and $V_2$, $V_1$ is a vertex cover, and hence it is a block deletion set.
Thus,  \textsc{Densest $k$-Subgraph} in $G'$ can be computed in time $O(2^{\db(G)/2} ( (k^3+\db(G))n + m))$.
Consequently, the total running time of this algorithm is  $O(2^{\db(G)/2} ( (k^3+\db(G))n + m))$.
\end{proof}

We can compute a minimum twin cover in time $O(1.2738^{\tc(G)}+\tc(G) n  + m)$ ~\cite{Ganian2015}.
Since $\db(G)\le \tc(G)$ and $1.2738<\sqrt{2}$, we obtain Theorem \ref{thm:twincover:DkS}.

\medskip

\revised{By using Theorem \ref{thm:bounded-cw}, Theorem \ref{thm:approx:DkS} can be generalized to a $\mathcal{C}$-deletion set where $\mathcal{C}$ is a class of bounded clique-width graphs.
\begin{theorem}
Given a $\mathcal{C}$-deletion set $D$, there is a $2$-approximation algorithm for \textsc{Densest $k$-Subgraph} that runs in time  $2^{|D|/2} n^{O(1)}$.
\end{theorem}
}

\section{Conclusion}\label{sec:conclusion}
In this paper, we studied the fixed-parameter tractability of \textsc{Densest $k$-Subgraph} by using structural parameters between clique-width and vertex cover.
We showed that \textsc{Densest $k$-Subgraph} is fixed-parameter tractable parameterized by neighborhood diversity, block deletion number, distance-hereditary deletion number, and cograph deletion number, respectively.
These results also hold for \textsc{Sparsest $k$-Subgraph} and \textsc{Maximum $k$-Vertex Cover}.
Moreover, we designed a 2-approximation  $2^{\tc(G)/2}n^{O(1)}$-time algorithm for  \textsc{Densest $k$-Subgraph}.
This improved a 3-approximation $2^{\vc(G)/2}n^{O(1)}$-time algorithm proposed in \cite{Bourgeois2017}.

  
As for future work, it is worth investigating the parameterized complexity for other structural parameters.
\revised{In particular, one of the most notable open questions would be the parameterized complexity for modular-width.}
\revised{Also, it might be interesting to discuss whether there
is a faster algorithm paramerized by neighborhood diversity without going through
quadratic integer program.}
  


\bibliographystyle{plain}
\bibliography{ref}
\end{document}